\documentclass{article}
\usepackage{ijcai22}

\usepackage{times}
\usepackage{soul}
\usepackage{url}
\usepackage[hidelinks]{hyperref}
\usepackage[utf8]{inputenc}
\usepackage[small]{caption}
\usepackage{graphicx}
\usepackage{amsmath}
\usepackage{amsthm}
\usepackage{booktabs}
\usepackage{algorithm}
\usepackage{algorithmic}
\usepackage{bm}
\usepackage{amsfonts}
\usepackage{enumerate}
\usepackage{multirow}
\usepackage{multicol}
\usepackage{caption}
\usepackage{subcaption}
\usepackage{color}
\usepackage{ising_v1}
\def\submission{1}

\newtheorem{definition}{Definition}[section]

\newtheorem{theorem}{Theorem}

\title{Practical and Secure Federated Recommendation with Personalized Mask}

\author{
Liu Yang$^{1,2}$\and
Junxue Zhang$^{1,2}$\and
Di Chai$^{1,2}$\and
Leye Wang$^{3}$\and
Kun Guo$^{4}$\and
Kai Chen$^{1}$\And
Qiang Yang$^{1}$
\affiliations
$^1$Hong Kong University of Science and Technology\\
$^2$Clustar\\
$^3$Peking University\\
$^4$Fuzhou University\\
\emails
\{lyangau, jzhangcs, dchai, kaichen, qyang\}@cse.ust.hk, leyewang@pku.edu.cn, gukn@fzu.edu.cn
}

\begin{document}

\maketitle

\begin{abstract}
Federated recommendation addresses the data silo and privacy problems altogether for recommender systems. Current federated recommender systems mainly utilize cryptographic or obfuscation methods to protect the original ratings from leakage. However, the former comes with extra communication and computation costs, and the latter damages model accuracy. Neither of them could simultaneously satisfy the real-time feedback and accurate personalization requirements of recommender systems. In this paper, we proposed federated masked matrix factorization (FedMMF) to protect the data privacy in federated recommender systems without sacrificing efficiency and effectiveness. In more details, we introduce the new idea of personalized mask generated only from local data and apply it in FedMMF. On the one hand, personalized mask offers protection for participants' private data without effectiveness loss. On the other hand, combined with the adaptive secure aggregation protocol, personalized mask could further improve efficiency. Theoretically, we provide security analysis for personalized mask. Empirically, we also show the superiority of the designed model on different real-world data sets. 
\end{abstract}

\section{Introduction}

\begin{figure}[t]
    \centering
    \includegraphics[width=1\linewidth]{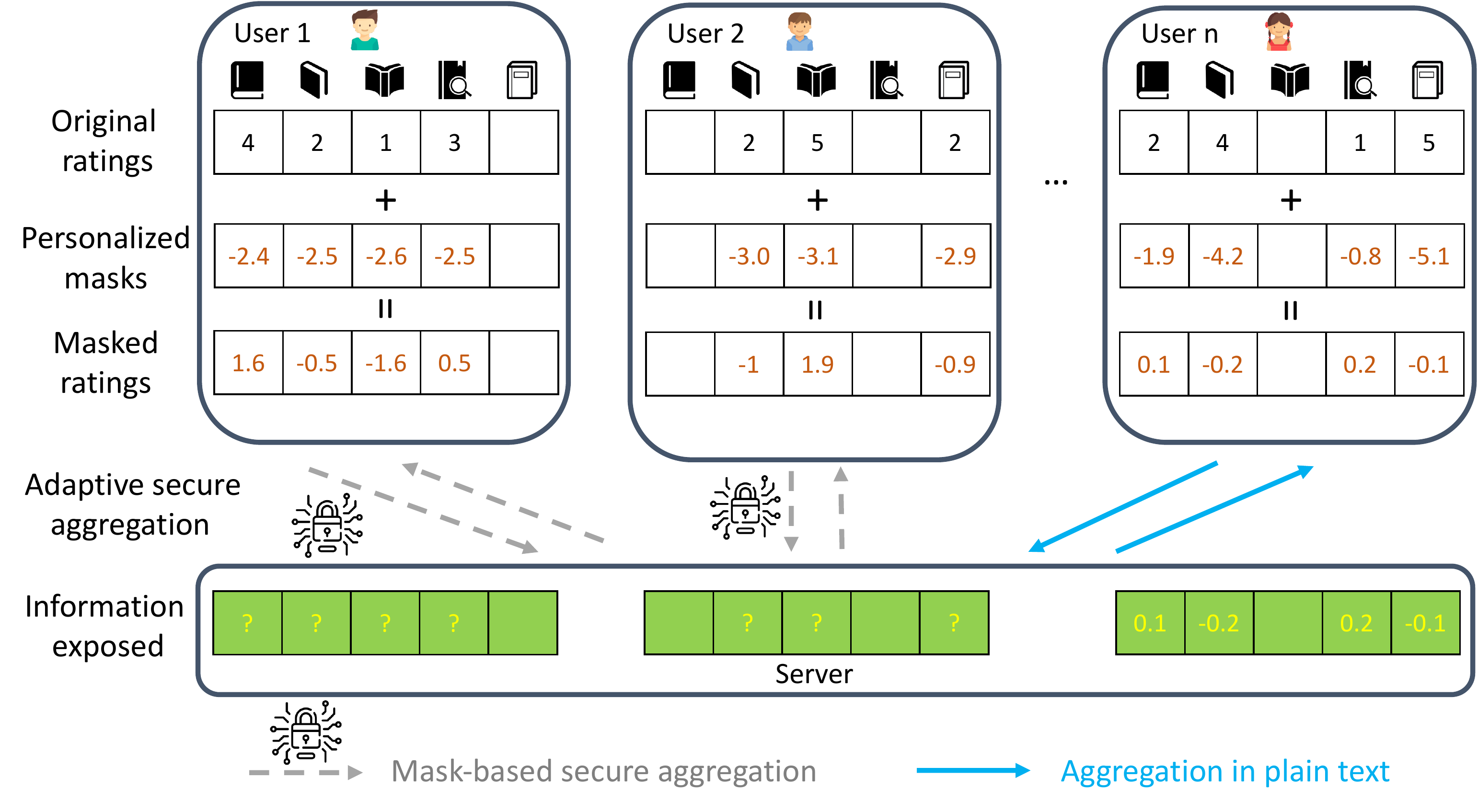}
    \caption{Illustration of the proposed FedMMF method. First, each party generates personalized masks via training a local model. Second, masked ratings are constructed via a combination of original ratings and personalized masks. Then, federated matrix factorization is performed on the masked ratings of all parties. An adaptive secure aggregation method is adopted. The parties with well-protected original data could share model updates via vanilla aggregation in plain-text format. And the other parties carry out a mask-based secure aggregation protocol. Finally, only the masked ratings with limited information are exposed to the server, leaking no data privacy.
    }
    \label{fig:illustration}
\end{figure}

Federated recommender system (FedRec) is an essential application of federated learning in the recommendation scenario~\cite{yang2020federated}. In recent years, federated learning has been a fast-growing research field, which keeps private data locally at multiple parties and trains models collaboratively in a secure and privacy-preserving way~\cite{mcmahan2021advances,mcmahan2017communication,yang2019federated}. For example,~\cite{ammad2019federated} proposed a federated matrix factorization algorithm, which distributes the training process at each local party and aggregates the computed gradients on the central server. Privacy-preserving is one of the major challenges in federated learning. Data decentralization does alleviate privacy risks compared with the conventional data-center training scheme. However, the gradients transmitted among different parties could still leak user privacy~\cite{aono2017privacy,zhu2019deep}. 

To address the privacy problem, current FedRec methods can be broadly divided into two categories. The first-kind solutions are based on cryptographic techniques such as homomorphic encryption (HE)~\cite{gentry2009fully} or secure multi-party computation (SMC)~\cite{yao1982protocols}. For example, HE-based FedRec~\cite{chai2020secure} utilizes HE to protect the transmitted gradients. These methods could lead to lossless model performance. However, they produce extra computation and communication costs since federated learning needs a large amount of calculation and intermediate results exchange. The second-kind solutions utilize the obfuscation methods such as differential privacy (DP)~\cite{dwork2014algorithmic}. For instance, DP-based FedRec~\cite{hua2015differentially} has been designed to provide a recommendation service without leaking the data privacy of multiple sources. Although DP-based federated algorithms are efficient, they damage the accuracy of models. Therefore, the above solutions all have difficulties when applying to practical problems. They cannot satisfy both the two requirements of recommender system (RecSys),~\ie, personalization and real-time. 

In this paper, we propose federated masked matrix factorization (FedMMF) as a novel FedRec method. The designed FedMMF method could protect the data privacy of FedRec without sacrificing efficiency and effectiveness. Shown in Fig.~\ref{fig:illustration}, instead of using cryptographic or obfuscation methods, we introduce a new idea of protecting private data from leakage in FedRec, which is called personalized mask. Personalized mask is a locally generated mask that adds on the original data for preserving privacy without effectiveness loss. Gradients computed on the masked ratings of one participant could be secure enough to directly share with other parties. Moreover, combined with the adaptive secure aggregation protocol, personalized mask also further relieves the efficiency problem of FedRec. Theoretically and empirically, we show the superiority of FedMMF. 




%

The paper is organized as follows, in Section~\ref{sec:preliminaries}, we first introduce the basic models and the privacy leakage problem; in Section~\ref{sec:fedmmf}, we explain the FedMMF algorithm, the training process, and the privacy guarantee; in Section~\ref{sec:experiment}, we show the performance of FedMMF in three real-world datasets.

\section{Preliminaries}
\label{sec:preliminaries}
In this section, we first introduce the traditional matrix factorization for recommendation. Then, based on the current challenges of RecSys, we explain federated matrix factorization (FedMF). Although FedMF alleviates the privacy problem of FedRec, there still exists leakage in the training process. Finally, we talk about the current solutions of secure FedMF.

\subsection{Matrix Factorization}
Given a rating matrix $\bm{R} \in \mathbb{R}^{n \times m}$, the recommender system aims to fill in the missing values of the matrix. Matrix factorization (MF) is regarded as one of the most classic recommendation algorithm~\cite{koren2009matrix}. It decouples the original matrix $\bm{R}$ into two low-rank matrices. The rating $r_{ui}$ that user $u$ gives to the item $i$ can be approximated as:

\begin{equation}
    \hat{r}_{ui} = \bm{q}^T_i \bm{p}_u,
\end{equation}
where $\bm{q}_i \in \mathbb{R}^{k \times 1}$ represents the latent factors of item $i$, $\bm{p}_u  \in \mathbb{R}^{k \times 1}$ represents the latent factors of user $u$, and the latent dimension $k$ can be regarded as the item's implicit characteristics. We could optimize the latent factors via minimizing the loss given below using the existing ratings:

\begin{equation}
    \mathop{\min}_{\bm{q}_i^*, \bm{p}_u^*} \frac{1}{2} \sum_{(u, i) \in \mathcal{K}} (r_{ui} - \bm{q}^T_i \bm{p}_u)^2 + \lambda (\left\| \bm{q}_i \right\|_2^2 + \left\| \bm{p}_u \right\|_2^2 ),
\end{equation}
where $\mathcal{K}$ stands for the set of user-item pairs whose rating $r_{ui}$ is already known and $\lambda$ is the regularization coefficient. Stochastic gradient descent is utilized to update each parameter:

\begin{equation}
    \bm{q}_i \leftarrow \bm{q}_i - \gamma \cdot (\lambda \cdot \bm{q}_i - e_{ui} \cdot \bm{p}_u),
\end{equation}

\begin{equation}
    \bm{p}_u \leftarrow \bm{p}_u - \gamma \cdot (\lambda \cdot \bm{p}_u - e_{ui} \cdot \bm{q}_i),
\end{equation}
where $e_{ui} = r_{ui} - \bm{q}^T_i \bm{p}_u$ and $\gamma$ is the learning rate. Conventional recommender systems centrally collect users' private data and train MF algorithm on the server, which leads to immense privacy risks.

\subsection{Federated Matrix Factorization}
With the development of federated learning, federated recommender system (FedRec) was proposed to address the privacy and data silo problems in the recommendation scenarios~\cite{yang2020federated}. In this paper, we focus on the horizontal FedRec, where each party only contains the rating information of one individual user and the user's private data is not allowed to leave the local device. Federated matrix factorization (FedMF) was designed to train recommendation models in such a naturally distributed situation. In the vanilla FedMF algorithm~\cite{ammad2019federated}, all the item latent factors $\{ \bm{q}_i \}_{i \in \mathcal{I}}$ are maintained on the central server, while each user's latent factors $\bm{p}_u$ is kept on the local party. The training process is as follows and loops until the convergence of model parameters: 1) party $u$ downloads item $i$'s latent factors $\bm{q}_i$ from the server; 2) party $u$ updates user's latent factors $\bm{p}_u$ using private local data $\bm{r}_u$; 3) party $u$ computes the gradients of each item's latent factors $\bm{\eta}_{ui} = \lambda \cdot \bm{q}_i - e_{ui} \cdot \bm{p}_u$ with $\bm{r}_u$ and the updated $\bm{p}_u$; 4) party $u$ sends $\bm{\eta}_{ui}$ to server; 5) server aggregates the gradients $\sum_{u \in \mathcal{U}} \bm{\eta}_{ui}$ and updates $\bm{q}_i$.

\subsubsection{Privacy Leakage from Gradients in FedMF}
Vanilla FedMF makes sure that users' private data never leaves the local parties. However, the transmitted gradients could also lead to privacy leakage~\cite{chai2020secure}. From user $u$, the server continuously receives the gradients of the item $i$'s latent vector at step $t-1$ and step $t$:

\begin{equation}
\label{eq:grad_t-1}
    \bm{\eta}_{ui}^{t-1} = \lambda \cdot \bm{q}^{t-1}_i - e_{ui}^{t-1} \cdot \bm{p}_u^{t-1},
\end{equation}

\begin{equation}
\label{eq:grad_t}
    \bm{\eta}_{ui}^t = \lambda \cdot \bm{q}^t_i - e_{ui}^t \cdot \bm{p}_u^t,
\end{equation}
where $e_{ui}^{t-1} = r_{ui} - {\bm{q}^{t-1}_i}^T \bm{p}_u^{t-1}$ and $e_{ui}^t = r_{ui} - {\bm{q}^t_i}^T \bm{p}_u^t$. Besides, the server also knows the update rule of the latent vector of user $u$:

\begin{equation}
\label{eq:diff_p}
    \bm{p}_u^t = \bm{p}_u^{t-1} + \gamma \cdot \sum_{i \in \mathcal{K}_u} (\lambda \cdot \bm{p}_u^{t-1} - e_{ui}^t \cdot \bm{q}^t_i),
\end{equation}
where $\mathcal{K}_u$ stands for the set of items that user $u$ has rated. Obviously, only $\bm{p}_u^{t-1}$, $\bm{p}_u^t$ and $r_{ui}$ are unknown to the server. Combining equations~\ref{eq:grad_t-1}, \ref{eq:grad_t}, and \ref{eq:diff_p}, the server could solve the unknown variables~\cite{lazard2009thirty}. In this way, private raw ratings of each user are revealed. 

\subsubsection{Secure FedMF}

To address the gradient leakage problem of vanilla FedMF, a few secure FedMF algorithms have been proposed. For example, HE-based FedMF~\cite{chai2020secure} and DP-based FedMF~\cite{hua2015differentially}, respectively, utilize HE and DP to further preserve privacy. HE-based FedMF encrypts gradients of item latent factors with HE before transmitting them to the server. Then, the server performs secure aggregation on the encrypted gradients, updates item latent factors in ciphertext state, and distributes the new encrypted item latent factors to each user. In a similar way, DP-based FedMF adds noises to gradients before aggregation. However, the former one causes extra costs and the latter one results in accuracy losses.

\section{Federated Masked Matrix Factorization}
\label{sec:fedmmf}
In this section, we explain the proposed FedMMF method. First,  FedMMF adopts a new idea of the personalized mask and we analyze its security. Then, FedMMF applies an adaptive secure aggregation protocol according to different protection situations provided by personalized masks on various users.

\subsection{Personalized Mask}
We generate the personalized masks via private well-trained model separately at each party. As shown in Fig.~\ref{fig:illustration}, FedMMF applies the idea of personalized mask in the previous FedMF architecture. The whole training process is as follows. Firstly, before the federated training of latent factors, each local party $u$ trains a private local model using only the user's own data. The corresponding loss function is shown below:

\begin{equation}
    L_u = \frac{1}{|\mathcal{K}_u|} \sum_{i \in \mathcal{K}_u} (r_{ui} - f_u^{mask}(i))^2
\end{equation}
Without loss of generality, we define the private model of user $u$ as $f_u^{mask}$. Then, the model is used to give prediction $f_u^{mask}(i)$ on each user-item pair $u, i$, where $i \in \mathcal{K}_u$. The opposite of the prediction is regarded as the personalized mask. Finally, all parties collaboratively train a matrix factorization model on the masked rating:

\begin{equation}
\label{eq:masked_rating}
    r_{u, i}^{masked} = r_{u, i} - f_u^{mask}(i).    
\end{equation}

The prediction of FedMMF algorithm for one specific user-item pair $(u, i)$ is:

\begin{equation}
    \hat{r}_{ui} = \bm{q}^T_i \bm{p}_u + f_u^{mask}(i).
\end{equation}
The private model $f_u^{mask}$ could be an arbitrary model which only trains on the local data. The well-behaved private model at each local party could protect the privacy of original ratings. Thus, parties with well-behaved private models are able to directly share their gradients computed on the masked ratings. Theorem~\ref{the:1} provides us with how much privacy could be protected by personalized masks.





\subsubsection{Security Analysis}
The private model $f_u^{mask}$ aims to hide the information of $r_{ui} \in \mathcal{R}$, which is the rating that each user $u \in \mathcal{U}$ gives to item $i \in \mathcal{I}$. For user $u$, the training data of $f_u^{mask}$ is denoted by $\mathcal{Z}^l = \{ (i, r_{ui}) \}_{i \in \{ 1, ..., l \}}$. The training data is sampled from a joint distribution $P_{\mathcal{IR}}$. We assume $\mathcal{R} \in [0, 1]$.

\theoremstyle{definition}
\begin{definition}[Privacy indicator of personalized mask]
We define the private information exposed by one specific user $u$ after applying personalized masks as:
\begin{equation}
J(f_u^{mask}, P_{\mathcal{IR}}) = E_{(\mathcal{I}, \mathcal{R}) \sim P_{\mathcal{IR}}} [\| \mathcal{R} - f_u^{mask}(\mathcal{I}) \|^2].
\end{equation}
\end{definition}
With a smaller value of privacy indicator $J$, personalized mask could provide a better protection. If the local private model predicts more accurately, personalized masks will cover more information of the original ratings.

\begin{theorem}
\label{the:1}
Personalized mask is $(\epsilon, \delta)-private$ for user $u$ if there exists a function $n_{\mathcal{F}_u}: (0,1) \times (0,1) \rightarrow \mathbb{N}$. For any $\epsilon, \delta \in (0, 1)$ and any distribution $P_{\mathcal{IR}}$, if $n > n_{\mathcal{F}}$, then 
\begin{equation}
\begin{aligned}
     Pr_{\mathcal{Z}^n \sim P_{\mathcal{IR}}} (J(f_u^{mask}, P_{\mathcal{IR}}) \leq \mathop{\min}_{f_u \in \mathcal{F}_u} & J(f_u, P_{\mathcal{IR}}) + \epsilon ) \\
     & \geq 1-\delta.
\end{aligned}
\end{equation}
\end{theorem}
\begin{proof}
For any $f_u \in \mathcal{F}_u$, the privacy indicator of user $u$ calculated on the training sample $\mathcal{Z}^n$ is:

\begin{equation}
\label{eq:privacy_indicator}
    J(f_u, P_{\mathcal{IR}}^n) = \frac{1}{n} \sum_{j=1}^n \| \mathcal{R}_j - f_u(\mathcal{I}_j) \|^2.
\end{equation}
Each $\| \mathcal{R}_j - f_u(\mathcal{I}_j) \|^2$ is an independent random variable with mean $J(f_u, P_{\mathcal{IR}})$. We further assume that $\| \mathcal{R}_j - f_u(\mathcal{I}_j) \|^2 \in [0, 1]$. According to Hoeffding’s inequality~\footnote{https://en.wikipedia.org/wiki/Hoeffding's\_inequality}, we obtain:

\begin{equation}
    Pr_{\mathcal{Z}^n \sim P_{\mathcal{IR}}} (\lvert (f_u, P_{\mathcal{IR}}^n) - J(f_u, P_{\mathcal{IR}}) \rvert \geq \epsilon) \leq 2e^{-2n \epsilon^2}, 
\end{equation}
then we could get:

\begin{equation}
\begin{aligned}
     Pr_{\mathcal{Z}^n \sim P_{\mathcal{IR}}} ( \exists f_u \in \mathcal{F}_u, s.t. \lvert (f_u, & P_{\mathcal{IR}}^n) - J(f_u, P_{\mathcal{IR}}) \rvert \\
     & \geq \epsilon) \leq 2\lvert \mathcal{F}_u \rvert e^{-2n \epsilon^2}.
\end{aligned}
\end{equation}
This shows that if 

\begin{equation}
    n \geq \frac{\log (2 \lvert \mathcal{F}_u \rvert / \delta)}{2 \epsilon^2},
\end{equation}
then

\begin{equation}
\begin{aligned}
     Pr_{\mathcal{Z}^n \sim P_{\mathcal{IR}}} ( \lvert (f_u, P_{\mathcal{IR}}^n) - J(f_u, & P_{\mathcal{IR}}) \rvert \leq \epsilon, \\
     & \forall f_u \in \mathcal{F}_u) \geq 1 - \delta,
\end{aligned}
\end{equation}
which is equivalent to:

\begin{equation}
\begin{aligned}
     Pr_{\mathcal{Z}^n \sim P_{\mathcal{IR}}} ( J(f_u^{mask}, P_{\mathcal{IR}}) \leq \mathop{\min}_{f_u \in \mathcal{F}} J(f_u, & P_{\mathcal{IR}}) + 2 \epsilon) \\
     & \geq 1 - \delta.
\end{aligned}
\end{equation}
The reason is that, given 

\begin{equation}
    \forall f_u \in \mathcal{F}_u, \lvert (f_u, P_{\mathcal{IR}}^n) - J(f_u, P_{\mathcal{IR}}) \rvert \leq \epsilon,
\end{equation}
we could obtain step by step:

\begin{equation}
\begin{aligned}
    J(f_u^{mask}, P_{\mathcal{IR}}) & \leq J(f_u^{mask}, P_{\mathcal{IR}}^n) + \epsilon \\
    & \leq \mathop{\min}_{f_u \in \mathcal{F}} J(f_u, P_{\mathcal{IR}}^n) + \epsilon \\
    & \leq \mathop{\min}_{f_u \in \mathcal{F}} J(f_u, P_{\mathcal{IR}}) + \epsilon + \epsilon \\
    & = \mathop{\min}_{f_u \in \mathcal{F}} J(f_u, P_{\mathcal{IR}}) + 2 \epsilon.
\end{aligned}
\end{equation}
Let $\epsilon = \frac{\epsilon}{2}$, we finally get

\begin{equation}
    n_{\mathcal{F}_u}(\epsilon, \delta) \leq \frac{2 \log (2 \lvert \mathcal{F}_u \rvert / \delta)}{2 \epsilon^2}.
\end{equation}
\end{proof}

The function $n_{\mathcal{F}_u}$ determines the sample complexity of user $u$ for training a FedMMF algorithm. It stands for how many samples at least are required by personalized masks to guarantee the privacy of user $u$. Besides, we assume the hypothesis class $\mathcal{F}_u$ of local private model is finite. However, it is not a necessary condition, and Theorem~\ref{the:1} can be further generalized. From Theorem~\ref{the:1}, we know that the privacy-preserving ability of personalized mask decides on the quality of local training data. The users with good enough local data could generate secure enough personalized masks, which successfully limit the exposed information from the masked ratings. The privacy indicator $J$ can be used to judge if the personalized masks are secure enough. In addition, we should also try to find the most suitable hypothesis class $\mathcal{F}_u$ on various data sets.


\begin{algorithm}[tb]
\caption{Federated Masked Matrix Factorization}
\label{alg:fedmmf}
\begin{algorithmic}[1]
    \STATE \textbf{Input}: $\bm{r}_{u \in \{ 1,...,n\}}, th_J$
    \STATE \textbf{Output}: $\bm{q}_{i \in \{ 1,...,m \}}, \bm{p}_{u \in \{ 1,...,n\}}, f_{u \in \{ 1,...,n\}}^{mask}$
    \STATE Server initializes $\bm{q}_{i \in \{ 1,...,m\}}^0$, each party $u$ initializes $\bm{p}_{u \in \{ 1,...,n\}}^0$ and $f_{u \in \{ 1,...,n\}}^{mask}(\theta_u)$.
    
    \item[]
    \FOR{each party $u \in \{ 1,...,n\}$ in parallel}
    \STATE // run on each party $u$
    \STATE Train private model $f_u ^{mask}(\theta_u)$ on local data $\bm{r}_u$;
    \STATE Compute personalized masked rating $r_{u, i}^{masked}$ according to Eq.~\ref{eq:masked_rating} for each $i \in \mathcal{K}_u$;
    \STATE Grouped to $\mathcal{U}_{secure}$ or $\mathcal{U}_{insecure}$ with $th_J$;
    \ENDFOR
    
    \item[]
    \STATE // run on the server
    \FOR{each $t = 1, 2, ..., T$}
    \FOR{each party $u \in \mathcal{U}_{secure}$ in parallel}
    \STATE Get gradients $\bm{\eta}_{ui \in \mathcal{K}_u}^t =$ {\bf MaskedUpdate}($\bm{q}_{i \in \mathcal{K}_u}^{t-1}$);
    \ENDFOR
    \FOR{each party $u \in \mathcal{U}_{insecure}$ in parallel}
    \STATE Get gradients $\bm{\tilde{\eta}}_{ui \in \mathcal{K}_u}^t =$ {\bf MaskedUpdate}($\bm{q}_{i \in \mathcal{K}_u}^{t-1}$);
    \ENDFOR

    \STATE Get the aggregated gradients $\sum_{u \in \mathcal{U}} \bm{\eta}_{i \in \mathcal{I}}^t$ according to Eq.~\ref{eq:adaptive_agg};
    \STATE Update item factors $\bm{q}_i^t = \bm{q}_i^{t-1} - \gamma \cdot \sum_{u \in \mathcal{U}} \bm{\eta}_i^t$ for each $i \in \mathcal{I}$;
    \ENDFOR
    
    \item[]
    \STATE // run on each party $u$
    \STATE {\bf MaksedUpdate:}
    \STATE Compute $e_{ui}^t = r_{ui}^{masked} - {\bm{q}^t_i}^T \bm{p}_u^t$ for each $i \in \mathcal{K}_u$;
    \STATE Update user factors $\bm{p}_u^t$ according to Eq.~\ref{eq:diff_p};
    \STATE Compute gradient $\bm{\eta}_{ui}^t$ according to Eq.~\ref{eq:grad_t} for each $i \in \mathcal{K}_u$;
    \STATE \textbf{Return} $\bm{\eta}_{ui \in \mathcal{K}_u}^{t}$ or $\bm{\tilde{\eta}}_{ui \in \mathcal{K}_u}^{t}$ to the server with adaptive secure aggregation protocol.
    
\end{algorithmic}
\end{algorithm}

\subsection{Adaptive Secure Aggregation}

The data quality of different users varies in the real world. Therefore, not all users can generate perfect personalized masks for protection. We propose an adaptive secure aggregation protocol to address this problem. For a given privacy indicator threshold $th_J$, we could divide the users into two groups,~\ie, secure masked group $\mathcal{U}_{secure}$ and insecure masked group $\mathcal{U}_{insecure}$. The privacy indicator $J$ of user in $\mathcal{U}_{secure}$ is larger than $th_J$, while the privacy indicator $J$ of user in $\mathcal{U}_{insecure}$ is smaller than $th_J$.

For user $u \in \mathcal{U}_{secure}$ with secure enough personalized masks, the gradients $\bm{\eta}_{ui}$ could be directly shared with the central server for aggregation. And the server could get $\sum_{u \in \mathcal{U}_{secure}} \bm{\eta}_{ui}$. However, for user $u \in \mathcal{U}_{insecure}$ with insecure personalized masks, sharing plain-text gradients will disclose the privacy of local rating data. Therefore, we adopt a mask-based secure aggregation method designed by ~\cite{bonawitz2017practical}. For an arbitrary pair of users $u, v \in \mathcal{U}_{insecure}$ and $u < v$, they decide a random mask $\bm{s}_{u, v} \in \mathbb{R}^{k \times 1}$ together. User $u$ adds this random mask $\bm{s}_{u, v}$ to its gradients, while user $v$ substracts $\bm{s}_{u, v}$ from its gradients. Then, each user $u$ could calculate:

\begin{equation}
	\bm{\tilde{\eta}}_{ui} = \bm{\eta}_{ui} + \sum_{u < v} \bm{s}_{u, v} - \sum_{u > v} \bm{s}_{v, u} \mod l,
\end{equation}
where $l$ is a large prime number. Next, each user $u \in \mathcal{U}_{insecure}$ sends the computed $\bm{\tilde{\eta}}_{ui}$ to the server. The server will calculate:

\begin{equation}
\begin{aligned}
	\sum_{u \in \mathcal{U}_{insecure}} \bm{\tilde{\eta}}_{ui} & = \sum_{u \in \mathcal{U}_{insecure}} (\bm{\eta}_{ui} + \sum_{u < v} \bm{s}_{u, v} - \sum_{u > v} \bm{s}_{v, u}) \\
	& = \sum_{u \in \mathcal{U}_{insecure}} \bm{\eta}_{ui} \mod l.
\end{aligned}
\end{equation}
The aggregated gradients could be obtained, and the gradients of one specific user are protected by the designed random masks. Furthermore, secret sharing~\cite{shamir1979share} is utilized to solve the dynamic user problem. With the adaptive secure aggregation protocol, the server could obtain 

\begin{equation}
\label{eq:adaptive_agg}
	\sum_{u \in \mathcal{U}} \bm{\eta}_{ui} = \sum_{u \in \mathcal{U}_{secure}} \bm{\eta}_{ui} + \sum_{u \in \mathcal{U}_{insecure}} \bm{\tilde{\eta}}_{ui}
\end{equation}
The details of FedMMF are shown in Algorithm~\ref{alg:fedmmf}. Compared with only applying the original aggregation in~\cite{bonawitz2017practical}, FedMMF utilizes the adaptive secure aggregation to further improve efficiency.

\section{Experiments}
\label{sec:experiment}
In this section, we show that FedMMF could improve efficiency without the loss of privacy and model effectiveness. Firstly, we explain the data sets, baseline models, and other settings in the experiments. Then, we show the improvements of FedMMF on model efficiency. With the help of adaptive secure aggregation protocol based on personalized masks, FedMMF accelerates the training process. At last, we discuss the model effectiveness of FedMMF with different kinds of personalized masks, compared to the baseline model.

\subsection{Settings}


We verify FedMMF on three real-world data sets. Two of them are MovieLens data sets~\cite{harper2015movielens}, \textit{i.e.}, MovieLens 100K and MovieLens 10M. The other one is the LastFM data set~\cite{Cantador:RecSys2011}. In our experiment, each user is regarded as a participant in the collaborative training process. Therefore, the user's own ratings are kept on the local party. Besides, we utilize the side information (\textit{i.e.}, user profiles and item attributes) to train the local private model. To construct features from tags in the data set, we utilize TFIDF~\cite{robertson2004understanding} and PCA~\cite{abdi2010principal} techniques. Besides, we set bins for the listening counts of music of the LastFM data set and convert them into ratings scaling from 1 to 5. In addition, the evaluation metrics of model efficacy are root mean square error (RMSE) and mean absolute error (MAE). They are averaged by each user-item pair but not each user, which is an alignment with most current works. Besides, we run each experiment ten times to obtain the mean and standard deviation values.




The compared models are: 1) {\bf FedMF}: parties collaboratively train matrix factorization models via sharing the latent factors of common users, where neither HE nor DP is utilized; 2) {\bf One-order FedMMF}: each party locally learns linear personalized masks to hide private rating information via a linear regression model~\cite{montgomery2012introduction}. Then, all parties collaboratively train FedMF on the one-order masked ratings; 3) {\bf Two-order FedMMF}: similarly, each party constructs two-order masks to protect private ratings via locally learning a factorization machine model~\cite{rendle2010factorization}; 4) {\bf High-order FedMMF}: each party captures high-order and nonlinear feature interactions through a neural network model~\cite{yegnanarayana2009artificial}. We do not compare FedMMF with DP-based FedMF, because DP causes effectiveness loss while FedMMF does not. Besides, we also show the performance of various local context models and federated context models for reference.

\begin{figure}[t]
     \centering
     \begin{subfigure}[b]{0.48\textwidth}
         \centering
    	\includegraphics[width=\textwidth]{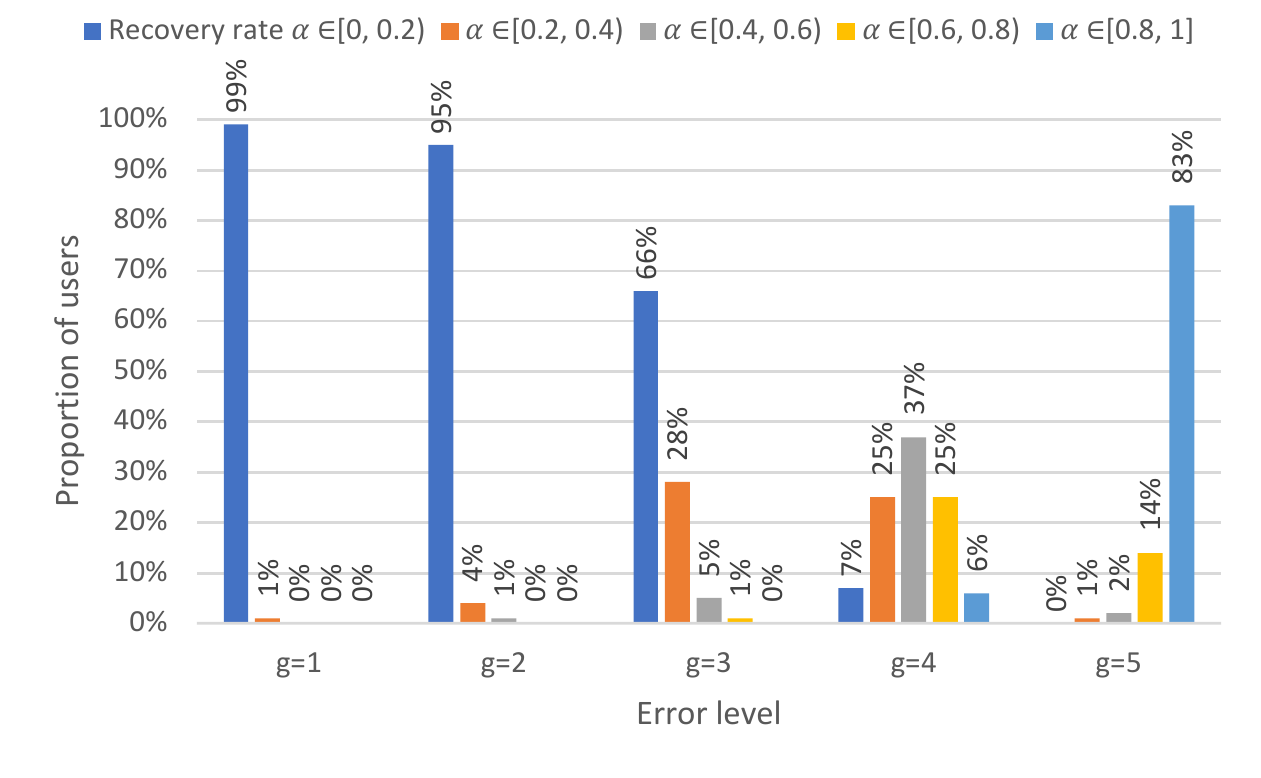}
         \caption{The results of recovery attack under different error levels. On the one hand, when error level $g$ is small, the recovery attack could hardly reveal the original rating information. On the other hand, when error level $g$ grows, the recovery attack becomes more accurate. However, the utility of recovered ratings also decreases.}
         \label{fig:recovery_attack}
     \end{subfigure}
     \hfill
     \begin{subfigure}[b]{0.48\textwidth}
         \centering
         \includegraphics[width=\textwidth]{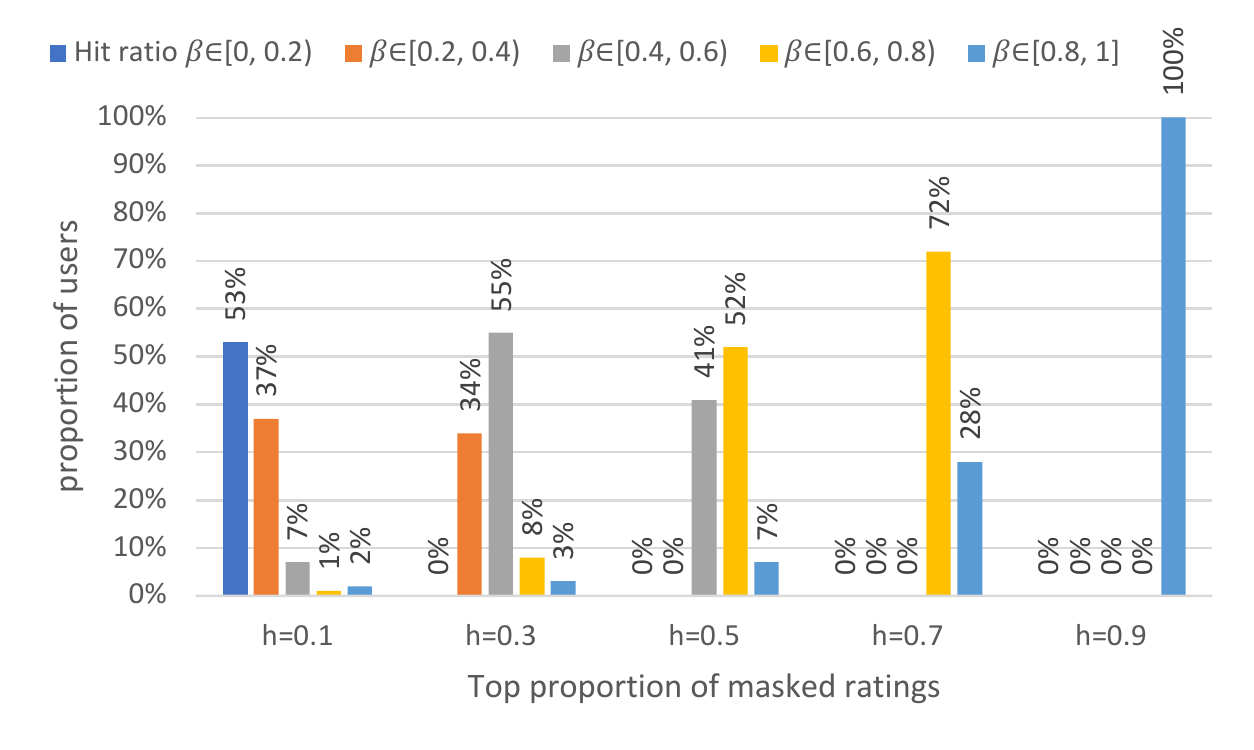}
         \caption{The results of ranking attack under different top proportions. With the masked ratings of each party, the adversary wants to choose the actual high-rated items. When the adversary utilizes a small top proportion $h$, the attacks performs on most parties achieve a poor hit ratio, which is less than 0.5. Although the hit ratio grows as $h$ increases, a large $h$ results in a useless ranking attack.}
         \label{fig:ranking_attack}
     \end{subfigure}
    \caption{Experiments results on recovery and ranking attacks.}
    \label{fig:attack_experiments}
\end{figure}

\subsection{Efficiency Promotion and Privacy Discussion}

Compared with HE-based FedMF~\cite{chai2020secure}, FedMMF with all users in the insecure user group could largely speed up the training process~\cite{bonawitz2017practical}. Then, the personalized mask technique could further improve the efficiency of the secure aggregation process via sharing plain-text gradients of parties with well-protected ratings. We provide two attack methods,~\ie, recovery attack and ranking attack, for analyzing how much the personalized mask technique could further promote model efficiency. Taking two-order FedMMF on MovieLens 10M data set as an example, we conduct the attack experiments. The rating range of the MovieLens 10M data set is from 0.5 to 5.0. And the rating interval is 0.5.


\begin{table*}[t]
\centering
\resizebox{\linewidth}{!}{
\begin{tabular}{lcccccc}
\toprule
\multirow{2}*{Models} & \multicolumn{2}{c}{MovienLens 100K} & \multicolumn{2}{c}{MovienLens 10M} & \multicolumn{2}{c}{LastFM} \\
~ & RMSE & MAE & RMSE & MAE & RMSE & MAE \\
\midrule
FedMF       & 0.9491 $\pm$ 0.0040 & 0.7412 $\pm$ 0.0027 & 0.7753 $\pm$ 0.0034 & 0.5827 $\pm$ 0.0015 & 1.2235 $\pm$ 0.0068 & 0.8780 $\pm$ 0.0047 \\
\midrule
LocalLR     & 1.0107 $\pm$ 0.0025 & 0.8040 $\pm$ 0.0022 & 0.8818 $\pm$ 0.0023 & 0.6766 $\pm$ 0.0011 & 1.1081 $\pm$ 0.0099 & 0.8163 $\pm$ 0.0085 \\

FedLR       & 1.0796 $\pm$ 0.0081 & 0.8844 $\pm$ 0.0058 & 0.9703 $\pm$ 0.0020 & 0.7497 $\pm$ 0.0012 & 1.5448 $\pm$ 0.0110 & 1.3538 $\pm$ 0.0137 \\

One-order FedMMF  & 0.9340 $\pm$ 0.0043 & 0.7340 $\pm$ 0.0035 & 0.7695 $\pm$ 0.0013 & 0.5808 $\pm$ 0.0008 & 1.0886 $\pm$ 0.0109 & 0.8066 $\pm$ 0.0092 \\
\midrule
LocalFM     & 1.0083 $\pm$ 0.0019 & 0.8054 $\pm$ 0.0019 & 0.8938 $\pm$ 0.0023 & 0.6862 $\pm$ 0.0012 & 1.0845 $\pm$ 0.0130 & 0.7988 $\pm$ 0.0035 \\

FedFM       & 1.0628 $\pm$ 0.0070 & 0.8644 $\pm$ 0.0053 & 0.9639 $\pm$ 0.0022 & 0.7445 $\pm$ 0.0015 & 1.5301 $\pm$ 0.0133 & 1.3369 $\pm$ 0.0117 \\

Two-order FedMMF  & \underline{0.9218 $\pm$ 0.0037} & \underline{0.7250 $\pm$ 0.0030} & 0.7720 $\pm$ 0.0013 & 0.5827 $\pm$ 0.0007 & \underline{1.0842 $\pm$ 0.0090} & 0.7964 $\pm$ 0.0031 \\
\midrule
LocalNN     & 1.0114 $\pm$ 0.0021 & 0.8087 $\pm$ 0.0017 & 0.8819 $\pm$ 0.0020 & 0.6816 $\pm$ 0.0009 & 1.1007 $\pm$ 0.0068 & 0.8007 $\pm$ 0.0123 \\

FedNN       & 1.0945 $\pm$ 0.0074 & 0.9176 $\pm$ 0.0060 & 0.9756 $\pm$ 0.0024 & 0.7689 $\pm$ 0.0028 & 1.5461 $\pm$ 0.0060 & 1.3598 $\pm$ 0.0075 \\

High-order FedMMF  & 0.9319 $\pm$ 0.0025 & 0.7317 $\pm$ 0.0018 & \underline{0.7648 $\pm$ 0.0016} & \underline{0.5772 $\pm$ 0.0008} & 1.0860 $\pm$ 0.0055 & \underline{0.7933 $\pm$ 0.0070} \\
\bottomrule
\end{tabular}}
\caption{Performance of FedMMF compared with baseline models on different data sets. FedMMF models with different personalized masks have no effectiveness loss compared with FedMF in all data sets. Besides the comparison between FedMMF and FedMF, we also show that FedMMF outperforms local context models and federated context models.}
\label{tab:performance}
\end{table*}

\subsubsection{Recovery Attack}

Against the masked ratings, an adversary could conduct an intuitive attack to recover the original ratings. However, the attack could be difficult if only the masked ratings are exposed. Therefore, for each party, we assume that the adversary knows the minimum and maximum values of the original ratings. Then, the adversary could scale the masked ratings to the range of original ratings for recovery. We define $g$ as the error level. If the difference between one recovered value and the corresponding original rating is less than $\alpha$, the recovery is considered successful. Thus, there exists a recovery rate $\alpha$ for each party's masked rating. In Fig.~\ref{fig:recovery_attack}, we show the proportion of parties whose recovery rate is in a certain range under different error levels. As we can see, when the error level is small, \textit{e.g.}, $g = 1$ and $g = 2$, the adversary could nearly reveal no party's privacy with a recovery rate larger than 0.5. And as the error level increases, the recovery rate begins to grow. However, a higher error level means a more inaccurate recovery, and the utility of the recovered ratings is poorer.

\subsubsection{Ranking Attack}
Since the intuitive recovery attack seems not successful enough, we introduce another method named ranking attack. Instead of recovering the original concrete ratings, ranking attack tries to find the high-rating items from their masked ratings. First, for each party, the adversary ranks the rated items according to their masked ratings. Then, items in the top $h$ proportion of masked ratings are selected as the high-rating items. Similarly, given $h$, we also sort these items with regard to their original ratings as the true high-rating item set. Thus, we could evaluate the ranking attack with hit ratio $\beta$, which is calculated as the ratio that items selected using masked ratings are in the true high-rating item set. Fig.~\ref{fig:ranking_attack} shows that, under different top proportion $h$, the ranking attack could reveal the rating ranking privacy of parties. If the selected top proportion is small, \textit{e.g.}, $h = 1$ and $h = 2$, the attacks performed on most parties' masked ratings obtain a hit ratio less than 0.5. It means that more than half of the selected items do not have high ratings. When the adversary tunes $h$ larger, the attack becomes more effective. However, a large $h$ is relatively meaningless because the adversary does not want to choose all items to be high-rating in reality.

According to the experiment results of the above two attack methods, we find that a considerable number of users get their rating privacy well-protected with the help of personalized masks. These users can be put in the secure group and transfer their gradients in plain text. Therefore, the personalized mask could further accelerate the training process of federated recommendation. Besides federated learning, the personalized masked ratings of users in the secure group could also be centrally collected and used for training without privacy leakage. This operation is able to reduce the communication and computation costs once again.



\subsection{Discussion on Model Effectiveness}

In this section, we verify the effectiveness of FedMMF on three real-world data sets. We implement three private models to construct personalized masks with different properties: the one-order mask, two-order mask, and high-order mask. The performances of FedMMF with these three masks are shown in Tab.~\ref{tab:performance}. RMSE and MAE are both regression evaluation metrics. Smaller value stands for better model efficacy. As we can see, FedMMF models with different personalized masks have no effectiveness loss compared with FedMF in all data sets. 

Moreover, FedMMF even outperforms FedMF. The effectiveness improvements could be divided into two parts. The first part benefits from the ensemble training scheme of FedMMF. The incorporation of personalized masks utilizes the idea of ensemble learning to combine weak learners for a better generalization ability. The second part takes advantage of the side information utilized in the private model of FedMMF. In the recommendation scenario, feature interactions are important information to capture. Another observation is that, on all three data sets, two-order FedMMF and high-order FedMMF dominate alternatively. It means we should utilize cross features to construct personalized masks in the recommendation scenarios. We also compare FedMMF models with corresponding local context and federated context models, shown in Tab.~\ref{tab:performance}. Comparing FedMMF with different local context models and federated context models, we could see that FedMMF also outperforms both of them. This observation verifies the main contribution to the effectiveness improvement is the incorporation of ensemble learning. On the other hand of the shield, FedMMF can also be regarded as an excellent way to combine collaborative information and feature information.

\section{Conclusion}
In this paper, we provide a new idea of personalized masks to protect data privacy in federated learning, which neither slows the training process down nor damages model performance. Taking the recommendation scenario as an example, we apply it in the FedMMF algorithm. Combining with the adaptive secure aggregation protocol, FedMMF shows superiority theoretically and empirically. In our future work, we would like to extend personalized masks to more general federated learning tasks besides recommender systems and try to combine personalized masks with differential privacy theory.

\clearpage

\bibliographystyle{named}
\bibliography{fedmmf}

\end{document}